\documentclass[11pt]{article}

\usepackage[margin=1in]{geometry}
\usepackage{amsthm,amsmath,amsfonts,amssymb}

\usepackage{colonequals}
\usepackage{thm-restate}
\usepackage{braket} 

\usepackage[pagebackref]{hyperref} 
\hypersetup{
    bookmarksnumbered=true, 
    unicode=false, 
    pdfstartview={FitH}, 
    pdftitle={An optimal quantum algorithm for the oracle identification problem}, 
    pdfauthor={Robin Kothari}, 
    pdfsubject={}, 
    pdfcreator={}, 
    pdfproducer={}, 
    pdfkeywords={}, 
    pdfnewwindow=true, 
    colorlinks=true, 
    linkcolor=blue, 
    citecolor=blue, 
    filecolor=blue, 
    urlcolor=blue 
}

\renewcommand{\backref}[1]{}

\renewcommand{\backrefalt}[4]{%
\ifcase #1 %
\or 
[p.\ #2]%
\else 
[pp.\ #2]%
\fi}

\let\oldbibliography\thebibliography
\renewcommand{\thebibliography}[1]{%
  \oldbibliography{#1}%
  \setlength{\itemsep}{0pt}%
}

\usepackage{algorithm}

\usepackage{algpseudocode}

\newcommand*\Let[2]{\State #1 $\gets$ #2}

\algrenewcommand{\algorithmiccomment}[1]{\hfill $\triangleright$ #1}

\newtheorem{theorem}{Theorem}
\newtheorem{lemma}{Lemma}
\newtheorem{proposition}{Proposition}

\theoremstyle{definition}
\newtheorem{definition}{Definition}

\newcommand{\eq}[1]{\hyperref[eq:#1]{(\ref*{eq:#1})}}
\renewcommand{\sec}[1]{\hyperref[sec:#1]{Section~\ref*{sec:#1}}}
\newcommand{\thm}[1]{\hyperref[thm:#1]{Theorem~\ref*{thm:#1}}}
\newcommand{\lem}[1]{\hyperref[lem:#1]{Lemma~\ref*{lem:#1}}}
\newcommand{\prop}[1]{\hyperref[prop:#1]{Proposition~\ref*{prop:#1}}}
\newcommand{\cor}[1]{\hyperref[cor:#1]{Corollary~\ref*{cor:#1}}}
\newcommand{\fig}[1]{\hyperref[fig:#1]{Figure~\ref*{fig:#1}}}
\newcommand{\tab}[1]{\hyperref[tab:#1]{Table~\ref*{tab:#1}}}
\newcommand{\alg}[1]{\hyperref[alg:#1]{Algorithm~\ref*{alg:#1}}}
\newcommand{\app}[1]{\hyperref[app:#1]{Appendix~\ref*{app:#1}}}

\newcommand{\comment}[1]{}

\newcommand{\C}{{\mathcal{C}}}
\newcommand{\R}{{\mathbb{R}}}

\newcommand{\ceil}[1]{\lceil{#1}\rceil}

\renewcommand{\(}{\left(}
\renewcommand{\)}{\right)}
\newcommand{\defeq}{\colonequals}
\newcommand{\poly}{\mathop{\mathrm{poly}}}

\newcommand{\norm}[1]{\|{#1}\|}

\renewcommand{\>}{\rangle}
\newcommand{\<}{\langle}

\newcommand{\oip}{\textsc{oip}}

\newcommand{\MAJ}{\textsc{maj}}

\begin{document}


\title{An optimal quantum algorithm for the oracle identification problem}

\author{
\normalsize Robin Kothari\\ [.5ex]
\small David R.\ Cheriton School of Computer Science and \\ 
\small Institute for Quantum Computing, University of Waterloo \\
\small \texttt{rkothari@uwaterloo.ca}
}

\date{}
\maketitle

\begin{abstract}
In the oracle identification problem, we are given oracle access to an unknown $N$-bit string $x$ promised to belong to a known set $\C$ of size $M$ and our task is to identify $x$. We present a quantum algorithm for the problem that is optimal in its dependence on $N$ and $M$. Our algorithm considerably simplifies and improves the previous best algorithm due to Ambainis et al. Our algorithm also has applications in quantum learning theory, where it improves the complexity of exact learning with membership queries, resolving a conjecture of Hunziker et al.

The algorithm is based on ideas from classical learning theory and a new composition theorem for solutions of the filtered $\gamma_2$-norm semidefinite program, which characterizes quantum query complexity. Our composition theorem is quite general and allows us to compose quantum algorithms with input-dependent query complexities without incurring a logarithmic overhead for error reduction. As an application of the composition theorem, we remove all log factors from the best known quantum algorithm for Boolean matrix multiplication.
\end{abstract}

\section{Introduction}
\label{sec:intro}

Query complexity is a model of computation where quantum computers are provably better than classical computers. Some of the great breakthroughs of quantum algorithms have been conceived in this model (e.g., Grover's algorithm~\cite{Gro96}). Shor's factoring algorithm~\cite{Sho97} also essentially solves a query problem exponentially faster than any classical algorithm. In this paper we study the query complexity of the oracle identification problem, the very basic problem of completely determining a string given oracle access to it.

In the oracle identification problem, we are given an oracle for an unknown $N$-bit string $x$, which is promised to belong to a known set $\C \subseteq \{0,1\}^N$, and our task is to identify $x$ while minimizing the number of oracle queries. For a set $\C$, we denote this problem ${\oip}(\C)$. As usual, classical algorithms are given access to an oracle that outputs $x_i$ on input $i$, while quantum algorithms have access to a unitary $O_x$ that maps $|i,b\>$ to $|i,b \oplus x_i\>$ for $b \in \{0,1\}$. For a function $f: D \to E$, where $D \subseteq \{0,1\}^N$, let $Q(f)$ denote the bounded-error quantum query complexity of computing $f(x)$. The problem $\oip(\C)$ corresponds to computing the identity function $f(x)=x$ with $D = E =\C$.

For example, let $\C_N \defeq \{0,1\}^N$. Then the classical query complexity of $\oip(\C_N)$ is $N$, since every bit needs to be queried to completely learn $x$, even with bounded error. A surprising result of van Dam shows that $Q(\oip(\C_N)) = N/2 + O(\sqrt{N})$~\cite{vDam98}. As another example, consider the set $\C_\text{H1} = \{x: |x| = 1\}$, where $|x|$ denotes the Hamming weight of $x$. This corresponds to the search problem with 1 marked item and thus $Q(\oip(\C_\text{H1})) = \Theta(\sqrt{N})$~\cite{BBBV97,Gro96}.

Due to the generality of the problem, it has been studied in different contexts such as quantum query complexity \cite{AIK+04,AIK+07}, quantum machine learning~\cite{SG04,AS05,HMP+10} and post-quantum cryptography \cite{BZ13}. Several well-known problems are special cases of oracle identification, e.g., the search problem with one marked element \cite{Gro96}, the Bernstein-Vazirani problem \cite{BV97}, the oracle interrogation problem~\cite{vDam98} and hidden shift problems \cite{vDHI06}. For some applications, generic oracle identification algorithms are almost as good as algorithms tailored to the specific application \cite{CKOR13}. Consequently, the main result of this paper improves some of the upper bounds stated in \cite{CKOR13}.

Ambainis et al.\ \cite{AIK+04,AIK+07} studied the oracle identification problem in terms of $N$ and $M\defeq|\C|$. They exhibited algorithms whose query complexity is close to optimal in its dependence on $N$ and $M$. For a given $N$ and $M$, we say an oracle identification algorithm is optimal in terms of $N$ and $M$ if it solves all $N$-bit oracle identification problems with $|\C|=M$ making at most $Q$ queries and there exists some $N$-bit oracle identification problem with $|\C|=M$ that requires $\Omega(Q)$ queries. This does not, however, mean that the algorithm is optimal for each set $\C$ individually, since these two parameters do not completely determine the query complexity of the problem. For example, all oracle identification problems with $M=N$ can be solved with $O(\sqrt{N})$ queries, and this is optimal since this class includes the search problem with 1 marked item ($\C_\text{H1}$ above).  However there exists a set $\C$ of size $M=N$ with query complexity $\Theta(\log N)$, such as the set of all strings with arbitrary entries in the first $\log N$ bits and zeroes elsewhere.

Let $\oip(M,N)$ denote the set of oracle identification problems with $\C \subseteq  \{0,1\}^N$ and $|\C| = M$. Let the query complexity of $\oip(M,N)$ be the maximum 
query complexity of any problem in that set.  Then the classical query complexity of $\oip(M,N)$ is easy to characterize: 

\begin{proposition}\label{prop:classicalOIP}
The classical (bounded-error) query complexity of $\oip(M,N)$ is $\Theta(\min\{M,N\})$. 
\end{proposition}

For $M\leq N$, the upper bound follows from the observation that we can always eliminate at least one potential string in $\C$ with one query. For the lower bound, consider any subset of $\C_\text{H1}$ of size $M$. For $M > N$, the lower bound follows from any set $\C \supseteq \C_\text{H1}$ and the upper bound is trivial since any query problem can be solved with $N$ queries.

Now that the classical query complexity is settled, for the rest of the paper ``query complexity'' will always mean quantum query complexity.  When quantum queries are permitted, the $M\leq N$ case is fully understood. For a lower bound, we consider (as before) any subset of $\C_\text{H1}$ of size $M$, which is as hard as the search problem on $M$ bits and requires $\Omega(\sqrt{M})$ queries. For an upper bound, we can reduce this to the case of $M=N$ by selecting $M$ bits such that the strings in $\C$ are distinct when restricted to these bits. (A proof of this fact appears in \cite[Theorem 11]{CKOR13}.) Thus $Q(\oip(M,N)) \leq Q(\oip(M,M))$, which is $O(\sqrt{M})$ \cite[Theorem 3]{AIK+04}. In summary, we have the following.

\begin{proposition}
For $M\leq N$, $Q(\oip(M,N))= \Theta(\sqrt{M})$.
\end{proposition}

For the hard regime, where $M > N$, the best known lower and upper bounds are the following, from \cite[Theorem 2]{AIK+04} and \cite[Theorem 2]{AIK+07} respectively.

\begin{theorem}[\cite{AIK+04,AIK+07}]\label{thm:AmbainisOIP}
If $N < M \leq 2^{N^{d}}$ for some constant $d<1$, then $Q(\oip(M,N))= O(\sqrt{{N\log M}/{\log N}})$ and for all $M > N$, $Q(\oip(M,N))= \Omega(\sqrt{{N\log M}/{\log N}})$. 
\end{theorem}

When $M$ gets closer to $2^N$, their algorithm no longer gives nontrivial upper bounds. For example, if $M \geq 2^{N/\log N}$, their algorithm makes $O(N)$ queries. 
While not stated explicitly, an improved algorithm follows from the techniques of \cite[Theorem 6]{AIN+09}, but the improved algorithm also does not yield a nontrivial upper bound when $M \geq 2^{N/\log N}$.
 Ambainis et al. \cite{AIK+07} left open two problems, in increasing order of difficulty: to determine whether it is always possible to solve the oracle identification problem for $M=2^{o(N)}$ using $o(N)$ queries and to design a single algorithm that is optimal in the entire range of $M$.

In this paper we resolve both open problems by completely characterizing the quantum query complexity of the oracle identification problem in the full range $N < M \leq 2^N$. Our main result is the following: 

\begin{restatable}{theorem}{qOIP}\label{thm:quantumOIP}
For $N < M \leq 2^N$, $Q(\oip(M,N)) = \Theta\(\sqrt{\frac{N\log M}{\log({N}/{\log M})+1}}\)$.
\end{restatable}

The lower bound follows from the ideas in \cite{AIK+04}, but needs additional calculation. We provide a proof in \app{lb}. The lower bound also appears in an unpublished manuscript \cite[Remark 1]{AIN+09}. The $+1$ term in the denominator is relevant only when $M$ gets close to $2^N$; it ensures that the complexity is $\Theta(N)$ in that regime.

Our main result is the algorithm, which is quite different from and simpler than that of \cite{AIK+07}.  It is also optimal in the full range of $M$ as it makes $O\(\sqrt{\frac{N\log M}{\log({N}/{\log M})+1}}\)$ queries when $M \geq N$ and $O(\sqrt{M})$ queries when $M \leq N$.     Our algorithm has two main ingredients:

First, we use ideas from classical learning theory, where the oracle identification problem is studied as the problem of exact learning with membership queries \cite{Ang88}. In particular, our quantum algorithm is based on Heged\H{u}s' implementation of the halving algorithm \cite{Heg95}. Heged\H{u}s characterizes the number of queries needed to solve the classical oracle identification problem in terms of the ``extended teaching dimension'' of  $\C$. While we do not use that notion, we borrow some of the main ideas of the algorithm. This is further explained in \sec{upper}.

We now present a high-level overview of the algorithm. Say we know that the string in the black box, $x$, belongs to a set $S$. We can construct from $S$ a string $s$, known  as the ``majority string,'' which is 1 at position $i$ if at least half the strings in $S$ are 1 at position $i$. Importantly, for any $i$, the set of strings in $S$ that disagree with $s$ at position $i$ is at most half the size of $S$. Now we search for a disagreement between $x$ and $s$ using Grover's algorithm. If the algorithm finds no disagreement, then $x = s$. If it does, we have reduced the size of $S$ by a factor of 2. This gives an algorithm with query complexity $O(\sqrt{N}\log M)$, which is suboptimal. We improve the algorithm by taking advantage of two facts: first, that Grover's algorithm can find a disagreement faster if there are many disagreements to be found, and second, that there exists an order in which to find disagreements that reduces the size of $S$ as much as possible in each iteration.  The existence of such an order was shown by Heged\H{u}s \cite{Heg95}.

The second ingredient of our upper bound is a general composition theorem for solutions of the filtered $\gamma_2$-norm semidefinite program (SDP) introduced by Lee et al.\ \cite{LMR+11} that preserves input-dependent query complexities. We need such a result to resolve the following problem: Our algorithm consists of $k$ bounded-error quantum algorithms that must be run sequentially because each algorithm requires as input the output of the previous algorithm. Let the query complexities of the algorithms be $Q_1(x), Q_2(x), \ldots , Q_k(x)$ on input $x$.  If these were exact algorithms, we could merely run them one after the other, giving one algorithm's output to the next as input, to obtain an algorithm with worst-case query complexity $O(\max_x \sum_i Q_i(x))$. However, since these are bounded-error algorithms, we cannot guarantee that all $k$ algorithms will give the correct output with high probability.  One option is to apply standard error reduction, but this would yield an algorithm that makes $O(\max_x \sum_i Q_i(x) \log k)$ queries. Instead, we prove a general composition theorem for the filtered $\gamma_2$-norm SDP that gives us an algorithm that makes $O(\max_x \sum_i Q_i(x))$ queries, as if the algorithms had no error. A similar result is known for worst-case query complexity, but that gives a suboptimal upper bound of $O(\sum_i \max_x Q_i(x))$ queries. We prove this result in \sec{gamma}.

The oracle identification problem was also studied by At{\i}c{\i} and Servedio \cite{AS05}, who studied algorithms that are optimal for a given set $\C$. The query complexity of their algorithm depends on a combinatorial parameter of  $\C$, $\hat{\gamma}^\C$, which satisfies $2 \leq 1/\hat{\gamma}^\C \leq N+1$. They prove $Q(\oip(\C)) = O(\sqrt{1/\hat{\gamma}^\C}\log M \log\log M)$. Our algorithm for oracle identification, without modification, makes fewer queries than this bound. Our algorithm's query complexity is $O\(\sqrt{\frac{1/\hat{\gamma}^\C}{\log{1/\hat{\gamma}^\C}}}\log M\)$, which resolves a conjecture of Hunziker et al.\ \cite{HMP+10}. We prove this in \sec{quantumml}.

Our composition theorem can also be used to remove unneeded log factors from existing quantum query algorithms. As an example, we show how to improve the almost optimal Boolean matrix multiplication algorithm that requires $O(n\sqrt{l} \poly(\log n))$ queries \cite{JKM12}, where $n$ is the size of the matrices and $l$ is the sparsity of the output, to an algorithm with query complexity $O(n\sqrt{l})$. We show this in \sec{bmm}. We conclude with some discussion and open questions in \sec{open}.

\section{Oracle identification algorithm}
\label{sec:upper}

In this section we explain the ideas that go into our algorithm and prove its correctness.  We also prove the query upper bound assuming we can compose bounded-error quantum algorithms without incurring log factors, which we justify in \sec{gamma}.

Throughout this section, let $x \in \C$ be the string we are trying to identify. For any set $S \in \{0,1\}^N$, let  $\MAJ(S)$ be an $N$-bit string such that $\MAJ(S)_i$ is 1 if $|\{y\in S:y_i = 1\}| \geq |\{y\in S:y_i = 0\}|$ and 0 otherwise. In words, $\MAJ(S)_i$ is $b$ if the majority of strings in $S$ have bit $i$ equal to $b$. Note that the string $\MAJ(S)$ need not be a member of $S$. 
In this paper, all logarithms are base 2 and for any positive integer $k$, we define $[k] \defeq \{1,2,\ldots,k\}$.

\subsection{Basic halving algorithm}

We begin by describing a general learning strategy called the halving algorithm, attributed to Littlestone \cite{Lit88}.  Say we currently know that the oracle contains a string $x \in S\subseteq \C$. The halving algorithm tests if the oracle string $x$ is equal to $\MAJ(S)$. If it is equal, we have identified $x$; if not, we look for a bit at which they disagree. Having found such a bit $i$, we know that $x_i \neq \MAJ(S)_i$, and we may delete all strings in $S$ that are inconsistent with this. Since at most half the strings in $S$ disagree with $\MAJ(S)$ at any position, we have at least halved the number of potential strings.

To convert this into a quantum algorithm, we need a subroutine that tests if a given string $\MAJ(S)$ is equal to the oracle string $x$ and finds a disagreement otherwise.  This can be done by running Grover's algorithm on the bitwise $\textsc{xor}$ of $x$ and $\MAJ(S)$. This gives us the following simple algorithm.

\begin{algorithm}
  \caption{Basic halving algorithm
    \label{alg:halving}}
  \begin{algorithmic}[1]
    \Statex
      \Let{$S$}{$\C$} 
		\Repeat 
		\State{Search for a disagreement between $x$ and $\MAJ(S)$. If we find a disagreement, delete all inconsistent strings from $S$. If not, let $S \gets \{\MAJ(S)\}$.}
		\Until{$|S|=1$}
  \end{algorithmic}
\end{algorithm}

This algorithm always finds the unknown string $x$, since $S$ always contains $x$. The loop can run at most $\log M$ times, since each iteration cuts down the size of $S$ by a factor of 2. Grover's algorithm needs $O(\sqrt{N})$ queries, but it is a bounded-error algorithm.  For this section, let us assume that bounded-error algorithms can be treated like exact algorithms and need no error reduction. Assuming this, \alg{halving} makes $O(\sqrt{N}\log M)$ queries.

\subsection{Improved halving algorithm}

Even assuming free error reduction, \alg{halving} is not optimal. Primarily, this is because Grover's algorithm can find an index $i$ such that $x_i \neq \MAJ(S)_i$ faster if there are many such indices to be found, and \alg{halving} does not exploit this fact.
Given an $N$-bit binary string, we can find a 1 with $O(\sqrt{{N}/{K}})$ queries in expectation, where $K>0$ is the number of 1s in the string \cite{BBHT98}. Alternately, there is a variant of Grover's algorithm that finds the first 1 (from left to right, say) in the string in $O(\sqrt{p})$ queries in expectation where $p$ is the position of the first 1. This follows from the known $O(\sqrt{N})$ algorithm for finding the first 1 in a string of size $N$ \cite{DHHM06}, by running that algorithm on the first $2^k$ bits, for $k=1,2,\ldots, \log N$. We can now modify the previous algorithm to look for the first disagreement between $x$ and $\MAJ(S)$ instead of any disagreement.

\begin{algorithm}
  \caption{Improved halving algorithm
    \label{alg:halving2}}
  \begin{algorithmic}[1]
    \Statex
      \Let{$S$}{$\C$} 
		\Repeat 
		\State{Search for the first disagreement between $x$ and $\MAJ(S)$. If we find a disagreement, delete all inconsistent strings from $S$. If not, let $S \gets \{\MAJ(S)\}$.}
		\Until{$|S|=1$}
  \end{algorithmic}
\end{algorithm}

As before, the algorithm always finds the unknown string. To analyze the query complexity, let $r$ be the number of times the loop repeats and $p_1, p_2, \ldots, p_r$ be the positions of disagreement found. After the first run of the loop, since a disagreement is found at position $p_1$, we have learned the first $p_1$ bits of $x$; the first $p_1-1$ bits agree with $\MAJ(S)$, while bit $p_1$ disagrees with $\MAJ(S)$. Thus we are left with a set $S$ in which all strings agree on these $p_1$ bits.  For convenience, we can treat $S$ as a set of strings of length $N-p_1$ (instead of length $N$).  Each iteration  reduces the effective length of strings in $S$ by $p_i$, which gives $\sum_i p_i \leq N$, since there are at most $N$ bits to be learned.  As before, the loop can run at most $\log M$ times, thus $r \leq \log M$. Finally, let us assume again that these bounded-error search subroutines are exact. Then this algorithm requires $O(\sum_i \sqrt{p_i})$ queries, which is $O(\sqrt{N\log M})$, by the Cauchy--Schwarz inequality.

\subsection{Final algorithm}

While \alg{halving2} is an improvement over \alg{halving}, it is still not optimal. One reason is that sometimes a disagreement between the majority string and $x$ may eliminate more than half the possible strings. This observation can be exploited by finding disagreements in such a way as to maximize the reduction in size when a disagreement is found. This idea is due to  Heged\H{u}s \cite{Heg95}. 

To understand the basic idea, consider searching for a disagreement between $x$ and $\MAJ(S)$ classically. The most obvious strategy is to check if $x_1 = \MAJ(S)_1$, $x_2 = \MAJ(S)_2$, and so on until a disagreement is found.  This strategy makes more queries if the disagreement is found at a later position.  However, we could have chosen to examine the bits in any order. We would like the order to be such that if a disagreement is found at a later position, it cuts down the size of $S$ by a larger factor. Such an ordering would ensure that either we spend very few queries and achieve a factor-2 reduction right away, or we spend more queries but the size of $S$ goes down significantly. Heged\H{u}s shows that there is always a reordering of the bits that achieves this. The following lemma is similar to \cite[Lemma 3.2]{Heg95}, but we provide a proof for completeness.

\begin{lemma}
\label{lem:ordering}
For any $S \subseteq \{0,1\}^N$, there exists a string $s \in \{0,1\}^N$ and a permutation $\sigma$ on $N$, such that for any $p \in [N]$, $|S_p| \leq \frac{|S|}{\max\{2,p\}}$, where $S_p = \{y\in S: y_{\sigma(i)} = s_{\sigma(i)} \text{ for }1\leq i\leq p-1 \text{ and } y_{\sigma(p)} \neq s_{\sigma(p)}\}$, the set of strings in $S$ that agree with $s$ at $\sigma(1), \ldots , \sigma(p-1)$ and disagree with it at $\sigma(p)$.
\end{lemma}

\begin{proof}
We will construct the permutation $\sigma$ and string $s$ greedily, starting with the first position, $\sigma(1)$.  We choose this bit to be one that intuitively contains the most information, i.e., a bit for which the fraction of strings that agree with the majority is closest to 1/2.  This choice will make $|S_1|$ as large as possible.  More precisely, we choose $\sigma(1)$ to be any $j$ that maximizes $|\{y\in S: y_j \neq  \MAJ(S)_j\}|$. Then let $s_{\sigma(1)}$ be $\MAJ(S)_{\sigma(1)}$. 

In general, after having chosen $\sigma(1), \ldots , \sigma(k-1)$ and having defined $s$ on those bits, we choose $\sigma(k)$ to be the most informative bit assuming all previous bits have agreed with string $s$ on positions $\sigma(1), \ldots , \sigma(k-1)$.  This choice makes $|S_{k}|$ as large as possible.  
More precisely, define $\bar{S}_p = \{y \in S: y_{\sigma(i)} = s_{\sigma(i)} \text{ for all } 1 \leq i \leq p\}$.  We choose $\sigma(k)$  to be any bit $j$ that maximizes $|\{y \in \bar{S}_{k-1}: y_{j}\neq \MAJ(\bar{S}_{k-1})_j\}|$. 
Then let $s_{\sigma(k)}$ be $\MAJ(\bar{S}_{k-1})_{\sigma(k)}$.

This construction ensures that $|S_1| \geq |S_2| \geq \ldots \geq |S_N|$.  Since $\sigma(k)$ was chosen to maximize $|\{y \in \bar{S}_{k-1}: y_{j}\neq \MAJ(\bar{S}_{k-1})_j\}|$, we have $|S_k| = |\{y \in \bar{S}_{k-1}: y_{\sigma(k)}\neq \MAJ(\bar{S}_{k-1})_{\sigma(k)}\}| \geq|\{y \in \bar{S}_{k-1}: y_{\sigma(k+1)}\neq \MAJ(\bar{S}_{k-1})_{\sigma(k+1)}\}|$. The size of this set is at least $|\{y \in \bar{S}_k: y_{\sigma(k+1)}\neq \MAJ(\bar{S}_{k-1})_{\sigma(k+1)}\}|$, since $\bar{S}_{k} \subseteq \bar{S}_{k-1}$.  We do not know the value of $\MAJ(\bar{S}_{k-1})_{\sigma(k+1)}$ (e.g., it need not be equal to $s_{\sigma(k+1)}$), but we do know that it is either 0 or 1. So this term is at least $\min\{|\{y \in \bar{S}_k: y_{\sigma(k+1)}\neq 0\}|,|\{y \in \bar{S}_k: y_{\sigma(k+1)}\neq 1\}|\} = \min\{|\{y \in \bar{S}_k: y_{\sigma(k+1)}\neq s_{\sigma(k+1)}\}|,|\{y \in \bar{S}_k: y_{\sigma(k+1)} = s_{\sigma(k+1)}\}|\} = \min\{|S_{k+1}|,|\bar{S}_{k+1}|\}  = |S_{k+1}|$, where the last equality uses $|S_k| \leq |\bar{S}_k|$ for all $k$.

Finally, combining $|S_1| + \ldots + |S_p| \leq |S|$ with $|S_1| \geq |S_2| \geq \ldots \geq |S_p|$ gives us $|S_p| \leq  |S|/{p}$. Combining this with $|S_1| \leq  |S|/2$, which follows from the definition of $S_1$, yields the result.
\end{proof}

We can now state our final oracle identification algorithm.

\begin{algorithm}
  \caption{Final algorithm
    \label{alg:final}}
  \begin{algorithmic}[1]
    \Statex
      \Let{$S$}{$\C$} 
		\Repeat 
		\State{Let $\sigma$ and $s$ be as in \lem{ordering}. Search for the first (according to $\sigma$) disagreement between $x$ and $s$. If we find a disagreement, delete all inconsistent strings from $S$. If not, let $S \gets \{s\}$.}
		\Until{$|S|=1$}
  \end{algorithmic}
\end{algorithm}

As before, it is clear that this algorithm solves the problem.  Let us analyze the query complexity.  To compute the query complexity, let $r$ be the number of times the loop repeats. Let $p_1, p_2, \ldots, p_r$ be the positions of disagreement.  We have $\sum_{i=1}^r p_i \leq N$, as in \alg{halving2}.

Unlike the previous analysis, the bound $r \leq \log M$ can be loose, since the size of $S$ may reduce by a larger factor due to \lem{ordering}.  Instead, we know that each iteration reduces the set $S$ by a factor of $\max\{2,p_i\}$, which gives us $\prod_{i=1}^{r} \max\{2,p_i\} \leq M$. As before, we will assume the search subroutine is exact, which gives us a query upper bound of $O(\sum_{i=1}^{r} \sqrt{p_i})$, subject to the constraints $\sum_{i=1}^r p_i \leq N$ and $\prod_{i=1}^{r} \max\{2,p_i\} \leq M$.  We solve this optimization problem in \app{proofs} to obtain the following lemma.

\begin{restatable}{lemma}{opt}
\label{lem:opt}
Let $C(M,N)$ be the maximum value attained by $\sum_{i=1}^{r} \sqrt{p_i}$, subject to the constraints
$\sum_{i=1}^{r} p_i \leq N,$ $\prod_{i=1}^{r} \max\{2,p_i\} \leq M,$ $r \in [N]$ and $p_i \in [N]$ for all $i \in [r]$.
Then $C(M,N) = O\(\sqrt{\frac{N\log M}{\log({N}/{\log M})+1}}\)$ and $C(M,N) = O(\sqrt{M})$.
\end{restatable}

Thus \alg{final} achieves the upper bound claimed in \thm{quantumOIP}, under our assumption. We can now return to the assumption that the search subroutine is exact. Since it is not exact, we could reduce the error with logarithmic overhead.  However, it is usually unnecessary to incur this loss in quantum query algorithms.  In the next section we prove this and rigorously establish the query complexity of \alg{final}.

\section{Composition theorem for input-dependent query complexity}
\label{sec:gamma}

The primary aim of this section is to rigorously establish the query complexity of \alg{final}. Along the way, we will develop techniques that can be used more generally. Let us begin by describing what we would like to prove.
\alg{final} essentially consists of a loop repeated $r(x)$ times. We write $r(x)$ to make explicit its dependence on the input $x$.  The loop itself consists of running a variant of Grover's algorithm on $x$, based on information we have collected thus far about $x$.  Call these algorithms $A_1, A_2, \ldots, A_{r(x)}$. To be clear, $A_1$ is the algorithm that is run the first time the loop is executed, i.e., it looks for a disagreement under the assumption that $S = \C$. It produces an output $p_1(x)$, which is then used by $A_2$. $A_2$ looks for a disagreement assuming a modified set $S$, which is smaller than $\C$.  Let us say that in addition to $p_2(x)$, $A_2$ also outputs $p_1(x)$. This ensures that the output of $A_i$ completely describes all the information we have collected about $x$. Thus algorithm $A_{i+1}$ now only needs the output of $A_i$ to work correctly.

We can now view \alg{final} as a composition of $r(x)$ algorithms, $A_1, A_2, \ldots, A_{r(x)}$. It is a composition in  the sense that the output of one is required as the input of the next algorithm. We know that the expected query complexity of $A_i$ is $O(\sqrt{p_i(x)})$. If these algorithms were exact, then running them one after the other would yield an algorithm with expected query complexity $O(\sum_i \sqrt{p_i(x)})$. But since they are bounded error, this does not work.  However, if we consider their worst-case complexities, we can achieve this query complexity.  If we have $r$ algorithms $A_1, A_2, \ldots, A_r$ with worst-case query complexities $Q_i$, then there is a quantum algorithm that solves the composed problem with $O(\sum_i Q_i)$ queries. This is a remarkable property of quantum algorithms, which follows from the work of Lee et al.\ \cite{LMR+11}. We first discuss this simpler result before moving on to input-dependent query complexities.

\subsection{Composition theorem for worst-case query complexity}

We now show a composition theorem for solutions of the filtered $\gamma_2$-norm SDP, which implies a similar result for worst-case quantum query complexity. This follows from the work of Lee et al.\ \cite{LMR+11}, which we generalize in the next section.

As discussed in the introduction, let $D \subseteq \{0,1\}^N$, and consider functions that map $D \to E$. For any matrix $A$ indexed by elements of $D$, we define a quantity $\gamma(A)$. (To readers familiar with the notation of \cite{LMR+11}, this is the same as their $\gamma_2(A|\Delta)$.)

\begin{definition}
Let $A$ be a square matrix indexed by $D$. We define $\gamma({A})$ as the following program.
\begin{align}
\gamma({A}) & \defeq \min_{\{\ket{u_{x j}}, \ket{v_{y j}}\}} \max_{x \in D} \quad c(x)\\
\label{eq:constr1}
\text{subject to:}& \qquad \forall x \in D, \quad c(x) = \max \Big\{
\sum_j \norm{\ket{u_{xj}}}^2, \sum_j \norm{\ket{v_{xj}}}^2\Big\}\\
&\qquad \forall x,y \in D, \quad \sum_{j:x_j \neq y_j}  \<{u_{xj}}|{v_{yj}}\> = A_{xy}
\end{align}
\end{definition}

We use $\gamma(A)$ to refer to both the semidefinite program (SDP) above and its optimum value. For a function $f:D\to E$, let $F$ be its Gram matrix, defined as $F_{xy} = 1$ if $f(x) \neq f(y)$ and $F_{xy} = 0$ otherwise. Lee et al.\ showed that $Q(f) = \Theta(\gamma(J-F))$, where $J$ is the all-ones matrix. 

More generally, they showed that this SDP also upper bounds the quantum query complexity of state conversion. In the state conversion problem, we have to convert a given state $|s_x\>$ to $|t_x\>$. An explicit description of the states $|s_x\>$ and $|t_x\>$ is known for all $x \in D$, but we do not know the value of $x$. Since the query complexity of this task depends only on the Gram matrices of the starting and target states, define $S$ and $T$ by $S_{xy} = \<s_x|s_y\>$ and $T_{xy} = \<t_x|t_y\>$ for all $x,y \in D$.  Let $S \mapsto T$ denote the problem of converting states with Gram matrix $S$ to those with Gram matrix $T$.  If $F$ is the Gram matrix of a function $f$, then $J \mapsto F$ is the function evaluation problem.  Lee et al.\ showed that $Q(S \mapsto T) = O(\gamma(S-T))$, which generalizes $Q(f) = O(\gamma(J-F))$.

We now have the tools to prove the composition theorem for the filtered $\gamma_2$-norm SDP.

\begin{theorem}[\cite{LMR+11}]
\label{thm:worstcasecomp}
Let $f_0, f_1, \ldots, f_k$ be functions with Gram matrices $F_0, F_1, \ldots, F_k$. Let $C_1, C_2, \ldots, C_k$ be the optimum value of the SDPs for the state conversion problems  $F_0 \mapsto F_1, F_1 \mapsto F_2, \ldots , F_{k-1} \mapsto F_k$, i.e., for $i \in [k]$, $C_i = \gamma(F_{i-1} - F_i)$.  Then, $\gamma(F_0 - F_k) \leq \sum_{i=1}^k C_i$.
\end{theorem}

This does not appear explicitly in \cite{LMR+11}, but simply follows from the triangle inequality $\gamma(A+B) \leq \gamma(A)+\gamma(B)$ \cite[Lemma A.2]{LMR+11}.  From this we can also show an analogous theorem for quantum query complexity, which states $Q(F_0 \mapsto F_k) = O(\sum_{i=1}^k Q(F_{i-1} \mapsto F_i))$. We do not prove this claim as we do not need it in this paper.

For our application, we require a composition theorem similar to \thm{worstcasecomp}, but for input-dependent query complexity.  However, it is not even clear what this means a priori, since the value $\gamma(J-F)$ does not contain information about input-dependent complexities.  Indeed, the value is a single number and cannot contain such information.  However, the SDP does contain this information and we modify this framework to be able to access this.

For example, let $f$ be the find-first-one function, which outputs the smallest $i$ such that $x_i=1$ and outputs $N+1$ if $x=0^N$.  There is a quantum algorithm that solves this with $O(\sqrt{f(x)})$ queries in expectation. Furthermore, there is a feasible solution for the $\gamma(J-F)$ SDP with $c(x)=O(\sqrt{f(x)})$, where $c(x)$ is the function that appears in \eq{constr1}. This suggests that $c(x)$ gives us information about the $x$-dependent query complexity.  The same situation occurs when we consider the search problem with multiple marked items.  There is a feasible solution with $c(x) = O(\sqrt{N/K})$ for inputs with $K$ ones.  This function $c(x)$ will serve as our input-dependent cost measure.

\subsection{Cost functions}

\begin{definition}[Cost function]
Let $A$ be a square matrix indexed by $D$. We say $c:D \to \R$ is a feasible cost function for $\gamma({A})$ if there is a feasible solution of $\gamma({A})$ with values $c(x)$ in eq. \eq{constr1}. 
Let the set of all feasible cost functions for $\gamma(A)$ be denoted $\Gamma(A)$.
\end{definition}

Note that if $c$ is a feasible cost function for $\gamma(J-F)$, then $\max_x c(x)$ is an upper bound on the worst-case cost, $\gamma(J-F)$, which is exactly what we expect from an input-dependent cost. We can now prove an input-dependent analogue of \thm{worstcasecomp} with $c(x)$ playing the role of $\gamma(J-F)$.

\begin{theorem}
\label{thm:comp}
Let $f_0,f_1, \ldots , f_k$ be functions with Gram matrices $F_0,F_1, \ldots, F_k$. Let $c_1, c_2, \ldots, c_k$ be feasible cost functions for $\gamma(F_0-F_1), \gamma(F_1 - F_2), \ldots, \gamma(F_{k-1} - F_{k})$, i.e., for $i \in [k]$, $c_i \in \Gamma(F_{i-1} - F_i)$. Then there is a $c \in \Gamma(F_0-F_k)$ satisfying $c(x) \leq \sum_i c_i(x)$ for all $x \in D$.
\end{theorem}

As in the case of \thm{worstcasecomp}, this follows from an analogous triangle inequality.

\begin{lemma}\label{lem:triangle}
Let $A$ and $B$ be square matrices indexed by $D$. If $c_A \in \Gamma(A)$ and $c_B\in \Gamma(B)$, there exists a $c \in \Gamma(A+B)$ satisfying $c(x) \leq c_A(x) + c_B(x)$ for all $x \in D$.
\end{lemma}

\begin{proof}
Since $c_A \in \Gamma(A)$ and $c_B\in \Gamma(B)$, there exist vectors that satisfy the following constraints:
$\sum_{j:x_j \neq y_j}  \<{u^{A}_{xj}}|{v^{A}_{yj}}\> = (A)_{xy}$ with $c_A(x) = \max \{\sum_j \norm{\ket{u^{A}_{xj}}}^2, \sum_j \norm{\ket{v^{A}_{xj}}}^2\}$ and $\sum_{j:x_j \neq y_j}  \<{u^{B}_{xj}}|{v^{B}_{yj}}\> = (B)_{xy}$ with $c_B(x) = \max \{\sum_j \norm{\ket{u^{B}_{xj}}}^2, \sum_j \norm{\ket{v^{B}_{xj}}}^2\}$.

Now define $\ket{u_{x j}}= \ket{1}\ket{u^{A}_{x j}} + \ket{2}\ket{u^{B}_{xj}}$ and 
$\ket{v_{x j}}= \ket{1}\ket{v^{A}_{x j}} + \ket{2}\ket{v^{B}_{xj}}$.  We claim that these vectors are feasible for $\gamma(J-G)$. The constraints are satisfied since $\sum_{j:x_j \neq y_j}  \<{u_{xj}}|{v_{yj}}\> = \sum_{j:x_j \neq y_j}  \<{u^{A}_{xj}}|{v^{A}_{yj}}\> + \sum_{j:x_j \neq y_j}  \<{u^{B}_{xj}}|{v^{B}_{yj}}\> = (A)_{xy} + (B)_{xy} = (A+B)_{xy}$. The cost function for this solution, $c(x)$, is $\max \{\sum_j \norm{\ket{u_{xj}}}^2, \sum_j \norm{\ket{v_{xj}}}^2\}$, which gives $c(x) = \max \{ \sum_j \norm{\ket{u^{A}_{xj}}}^2+\norm{\ket{u^{B}_{xj}}}^2, \sum_j \norm{\ket{v^{A}_{xj}}}^2+\norm{\ket{v^{B}_{xj}}}^2 \} \leq c_A(x) + c_B(x)$.
\end{proof}

In our applications, we will encounter algorithms that also output their input, i.e., accept as input $f(x)$ and output $(f(x),g(x))$.  Note that the Gram matrix of the function $h(x) = (f(x),g(x))$ is merely $H = F \circ G$, defined as $H_{xy} = F_{xy} G_{xy}$.  

Such an algorithm can either be thought of as a single quantum algorithm that accepts $f(x)\in E$ as input and outputs $(f(x),g(x))$ or as a collection of algorithms $A_e$  for each $e\in E$, such that algorithm $A_{f(x)}$ requires no input and outputs $(f(x),g(x))$ on oracle input $x$.  These are equivalent viewpoints, since in one direction you can construct the algorithms $A_e$ from $A$ by hardcoding the value of $e$ and in the other direction, we can read the input $e$ and call the appropriate $A_e$ as a subroutine and output $(e, A_e(x))$. Additionally, if the algorithm $A_{f(x)}$ makes $q(x)$ queries on oracle input $x$, the algorithm $A$ we constructed accepts $f(x)$ as input, outputs $(f(x),g(x))$, and makes $q(x)$ queries on oracle input $x$. While intuitive for quantum algorithms, we need to establish this rigorously for cost functions.

\begin{theorem}\label{thm:circ}
Let $f,g:D \to E$ be functions with Gram matrices $F$ and $G$. For any $e \in E$, let $f^{-1}(e) = \{x: f(x)=e\}$. For every $e \in E$, let $c_e:f^{-1}(e)\to \R$ be a feasible cost function for $\gamma(J - G_e)$, where $G_e$ denotes the matrix $G$ restricted to those $x$ that satisfy $f(x) = e$. Then there exists a $c\in \Gamma(F - F\circ G)$, such that $c(x) = c_{f(x)}(x)$.
\end{theorem}

\begin{proof}
We build a feasible solution for $\gamma(F - F \circ G)$ out of the feasible solutions for $\gamma(J-G_e)$. We have vectors $\{\ket{u^e_{x j}}, \ket{v^e_{y j}}\}$ for each $e \in E$ that satisfy  $\sum_{j:x_j \neq y_j}  \<{u^e_{xj}}|{v^e_{yj}}\> = (J-G_e)_{xy}$ for all $x,y \in f^{-1}(e)$ and $c_e(x) = \max \{\sum_j \norm{\ket{u^e_{xj}}}^2, \sum_j \norm{\ket{v^e_{xj}}}^2\}$.

Let $\ket{u_{x j}}= \ket{f(x)}\ket{u^{f(x)}_{xj}}$ and $\ket{v_{x j}}= \ket{f(x)}\ket{v^{f(x)}_{xj}}$.  This is a feasible solution for $\gamma(F - F\circ G)$, since $\sum_{j:x_j \neq y_j}  \<{u_{xj}}|{v_{yj}}\>  = \sum_{j:x_j \neq y_j}  \<f(x)|f(y)\>\<{u^{f(x)}_{xj}}|{v^{f(y)}_{yj}}\> = F_{xy} \circ (J - G_{f(x)})_{xy} = F_{xy} - (F \circ G)_{xy}$.  Note that when $f(x) \neq f(y)$, the value of $\sum_{j:x_j \neq y_j}  \<{u^{f(x)}_{xj}}|{v^{f(y)}_{yj}}\>$ is not known, but this only happens when $F_{xy} = 0$, which makes the term 0.  Lastly, the cost function for this solution is $\max \{\sum_j \norm{\ket{u_{xj}}}^2, \sum_j \norm{\ket{v_{xj}}}^2\}$, which is $\max \{ \sum_j \norm{\ket{u^{f(x)}_{xj}}}^2, \sum_j \norm{\ket{v^{f(x)}_{xj}}}^2 \} = c_{f(x)}(x)$.
\end{proof}

\subsection{Algorithm analysis}

We can now return to computing the query complexity of \alg{final}. Using the same notation as in the beginning of this section, for any $x \in \C$, we define $r(x)$ to be the number of times the repeat loop is run in \alg{final} for oracle input $x$ assuming all subroutines have no error. Similarly, let $p_1(x),p_2(x),\ldots p_{r(x)}(x)$ be the first positions of disagreement found in each run of the loop. Note that $p_1(x),p_2(x),\ldots p_{r(x)}(x)$ together uniquely specify $x$. Let $r = \max_x r(x)$.

We now define $r$ functions $f_1, \ldots, f_r$ as $f_1(x) = p_1(x), f_2(x) = (p_1(x),p_2(x)), \ldots, f_r(x) = (p_1(x), \ldots, p_r(x))$, where $p_k(x) = 0$ if $k>r(x)$. Thus if $P_i$ are the Gram matrices of the functions $p_i$, then $F_1 = P_1, F_2 = P_1 \circ P_2, \ldots, F_r = P_1 \circ P_2 \circ \cdots \circ P_r$. 

We will now construct a solution for $\gamma(J-F_r)$, using solutions for the intermediate functions $f_i$.  From \thm{comp} we know that we only need to construct solutions for $\gamma(J-F_1), \gamma(F_1 - F_2), \ldots ,\gamma(F_{r-1} - F_r)$. From \thm{circ} we know that instead of constructing a solution for $\gamma(F_k - F_{k+1})$, which is $\gamma(F_k - F_k \circ P_{k+1})$, we can construct several solutions, one for each value of $f_k(x)$.  More precisely, let $f_k:D \to E_k$; then we can construct solutions for $\gamma(J - P_{k+1}^e)$ for all $e\in E_k$, where $P_{k+1}^e$ is the matrix $P_{k+1}$ restricted to $x$ that satisfy $f_k(x) = e$.

For any $k$, the problem corresponding to $\gamma(J - P_{k+1}^e)$ is just the problem of finding the first disagreement between $x$ and a known string, which is the essentially the find-first-one function.  This has a solution with cost function $O(\sqrt{f(x)})$, which in this case is $O(\sqrt{p_{k+1}(x)})$.

\begin{theorem}
\label{thm:firstmarked}
Let $f$ be the function that outputs the smallest $i$ such that $x_i = 1$ and outputs $N+1$ if $x = 0^N$ and let $F$ be its Gram matrix. Then there is a $c \in \Gamma(J-F)$ such that $c(x) = O(\sqrt{f(x)})$.
\end{theorem}

\begin{proof}
Let $a_k = k^{-1/4}$ and $b_k = 1/a_k = k^{1/4}$.  Define $|u_{xj}\> = |v_{xj}\>$ as the following.
\[
|u_{xj}\> = |v_{xj}\> = \begin{cases} a_j, & \text{if } j < f(x) \\ b_{f(x)}, & \text{if } j = f(x) \\ 0, & \text{if } j > f(x) .\end{cases}
\]
This is a feasible solution for $\gamma(J-F)$.  Since the constraints are symmetric in $x$ and $y$, there are two cases: either $f(x) < f(y)$ or $f(x) = f(y)$.  For the first case, $\sum_{j:x_j \neq y_j}  \<{u_{xj}}|{v_{yj}}\> = \sum_{j=f(x)}  \<{u_{xj}}|{v_{yj}}\> = a_{f(x)} b_{f(x)} = 1$, since $x$ and $y$ agree on all positions before $f(x)$.  For the second case, $\sum_{j:x_j \neq y_j}  \<{u_{xj}}|{v_{yj}}\> = 0$, since the only bits that $x$ and $y$ disagree on appear after position $f(x)=f(y)$.
To compute the cost function, note that $c(0^N) = \sum_{k=1}^{N} a_k^2 = O(\sqrt{N}) = O(\sqrt{f(0^N)})$. For all other $x$, $c(x) = \sum_{k=1}^{f(x)-1} a_k^2 + b_{f(x)}^2 = \sum_{k=1}^{f(x)-1} k^{-1/2} + \sqrt{f(x)} = O(\sqrt{f(x)})$.
\end{proof}

Our function is different from this one in two ways.  First, we wish to find the first disagreement with a fixed string $s$ instead of the first 1.  This change does not affect the Gram matrix or the SDP. Second, we are looking for a disagreement according to an order $\sigma$, not from left to right.  This is easy to fix, since we can replace $j$ with $\sigma(j)$ in the definition of the vectors in the proof above.

This shows that for any $k$, there is a feasible cost function for $\gamma(J - P_{k+1}^e)$ with cost $c(x)= O(\sqrt{p_{k+1}(x)})$ for any $x$ that satisfies $f_k(x) = e$. Using \thm{circ}, we get that for any $k$ there is a $c_k \in \Gamma(F_k - F_k\circ P_{k+1})$ with $c_k(x) = O(\sqrt{p_{k+1}(x)})$ for all $x \in D$.  Finally, using \thm{comp}, we have a $c \in \Gamma(J-F_r)$ with cost $c(x) = O(\sum_{i=1}^{r} \sqrt{p_i(x)}) = O(\sum_{i=1}^{r(x)} \sqrt{p_i(x)})$.

Since the function $f_r(x)$ uniquely determines $x$, we have a feasible cost function for oracle identification with cost $O(\sum_{i=1}^{r(x)} \sqrt{p_i(x)})$, subject to the constraints of \lem{opt}, which we have already solved. Along with the lower bound proved in \app{lb}, this yields the main result.

\qOIP*

\section{Other applications}

\subsection{Quantum learning theory}
\label{sec:quantumml}

The oracle identification problem has also been studied in quantum learning theory with the aim of characterizing $Q(\oip(\C))$. The algorithms and lower bounds studied apply to arbitrary sets $\C$, not just to the class of sets of a certain size, as in the rest of the paper. We show that \alg{final} also performs well for any set $\C$, outperforming the best known algorithm.  The known upper and lower bounds for this problem are in terms of a combinatorial parameter $\hat{\gamma}^\C$, defined by Servedio and Gortler. They showed that for any $\C$, $Q(\oip(\C))=\Omega(\sqrt{1/\hat{\gamma}^\C} + \frac{\log M}{\log N})$ \cite{SG04}. Later,   At{\i}c{\i} and Servedio showed that $Q(\oip(\C)) =O(\sqrt{1/\hat{\gamma}^\C} \log M \log\log M)$ \cite{AS05}. 

While we do not define $\hat{\gamma}^\C$, we can informally describe it as follows:  $\hat{\gamma}^\C$ is the largest $\alpha<1$, such that for any set $S\subseteq \C$, if we know that $x$ belongs to $S$, there is a bit of $x$ that can be queried such that size of the set of strings consistent with the answer to this query is at most $(1-\alpha)|S|$, no matter what the oracle responds. This ensures that if we query the oracle with the permutation of \lem{ordering}, which was chosen to maximize the number of strings eliminated with a query, each query reduces the size of $S$ by a factor of $(1-\hat{\gamma}^\C)$.

This adds an extra constraint to \lem{opt} of the form $M \prod_i^r (1-\hat{\gamma}^\C)^{p_i} \geq 1$, since learning $p_i$ bits will reduce the size of the remaining set by a factor of  $(1-\hat{\gamma}^\C)^{p_i}$. From this constraint we get $(\sum_i p_i) \log(1-\hat{\gamma}^\C) \geq -\log M$. Using $\log(1-\hat{\gamma}^\C) \leq -\hat{\gamma}^\C$ gives $\sum_i p_i \leq \frac{\log M}{\hat{\gamma}^\C}$.

We may now replace the constraint $\sum_i p_i \leq N$ with $\sum_i p_i \leq \frac{\log M}{\hat{\gamma}^\C}$ in the optimization problem of \lem{opt}.  This inequality also implies $p_i \leq \frac{\log M}{\hat{\gamma}^\C}$ and $r \leq \frac{\log M}{\hat{\gamma}^\C}$. Thus we may simply replace all occurrences of $N$ by $\frac{\log M}{\hat{\gamma}^\C}$ in \lem{opt}. This yields the following theorem, which resolves a conjecture of Hunziker et al.\ \cite[Conjecture 2]{HMP+10}.

\begin{theorem}
\alg{final} solves $\oip(\C)$ with $O\(\sqrt{\frac{1/\hat{\gamma}^\C}{\log {1/\hat{\gamma}^\C}}}\log M\)$ queries.
\end{theorem}

This shows that \alg{final} performs well on any set $\C$, since $Q(\oip(\C))=\Omega(\sqrt{1/\hat{\gamma}^\C} + \frac{\log M}{\log N})$. By combining this lower bound with our upper bound, we see that \alg{final} makes $O(\frac{Q(\oip(\C))^2}{\sqrt{\log Q(\oip(\C))}} \log N)$ queries, which means it can be at most about quadratically worse than the best algorithm for $\oip(\C)$.
 
\subsection{Boolean matrix multiplication}
\label{sec:bmm}

In this section we show how to improve the upper bound on Boolean matrix multiplication (BMM) from $O(n\sqrt{l} \poly(\log n))$ \cite{JKM12} to $O(n\sqrt{l})$, where $n$ is the size of the matrices and $l$ is the sparsity of the output. 
Just like in the oracle identification problem, we will break up the BMM algorithm of \cite{JKM12} into a sequence of algorithms $A_i$ such that the output of $A_i$ is the input of $A_{i+1}$, and convert each algorithm into a feasible solution for the corresponding SDP.

The BMM algorithm is almost of this form. The main algorithm uses two subroutines for graph collision, one to solve the decision problem and another to find all collisions. The first subroutine solves the decision problem on a bipartite graph with $2n$ vertices and $m$ nonedges in $O(\sqrt{n}+\sqrt{m})$ queries. Since the graph is not part of the oracle input, this query complexity is not input dependent, and thus there is a feasible SDP solution for this problem with $c(x) = O(\sqrt{n}+\sqrt{m})$ for all $x$, using the known characterization of Lee et al.\ \cite{LMR+11}.

The second subroutine finds all graph collisions in an instance with $\lambda$ collisions using $O(\sqrt{n\lambda} + \sqrt{m})$ queries. This upper bound is input dependent, since $\lambda$ is a function of the input. In this subroutine, the only input-dependent algorithm is the variant of Grover's algorithm that requires $O(\sqrt{nk})$ queries to output all the ones in an $n$-bit string when there are $k$ ones. It is easy to show that there is a feasible cost function for this with $c(x)=O(\sqrt{nk})$. For example, we may compose the SDP solution for the find-first-one function (\thm{firstmarked}) with itself repeatedly to find all ones. The cost function of the resultant SDP will satisfy $c(x) = O(\sum_i \sqrt{p_i})$, where $p_i$s are the locations of the ones. By the Cauchy--Schwarz inequality this is $O(\sqrt{nk})$. Thus the second graph collision algorithm also has a feasible cost function $c(x)=O(\sqrt{n\lambda} + \sqrt{m})$.

The BMM algorithm breaks up the problem into $n$ instances of graph collision. The algorithm repeatedly searches for indices $i$ such that the $i$th graph collision instance has a collision. Then it finds all graph collisions of this instance and repeats. Instead of searching for an arbitrary $i$, we can search for the first index $i$. The problem of searching for the first $i$ that has a graph collision is the composition of the find-first-one function (\thm{firstmarked}) with the graph collision function. This is a composition in the sense that each oracle input bit of the first problem is the output bit of another query problem.  It is known that the optimal value of the $\gamma$ SDP for $f \circ g^n$ is at most $\gamma(J-F)\gamma(J-G)$. Similarly, it can be shown that there is a feasible cost function for $f \circ g$ that is at most the product of the cost functions. This is similar to \cite[Lemma 5.1]{LMR+11} or \lem{triangle}, but instead of taking the direct sum of the vectors, we take the tensor product.

Finally, let $p_1, \ldots, p_t$ be the positions of indices found in the algorithm. The search problem requires $O(\sqrt{p_i}(\sqrt{n}+\sqrt{m}))$ queries for each $i$, since it is the composition of the two above-mentioned algorithms. The algorithm that finds all graph collisions has a feasible cost function $O(\sqrt{n\lambda_i} + \sqrt{m})$, where $\lambda_i$ is the number of graph collisions in the $i$th graph collision instance.  This gives a feasible cost function for BMM with cost $O(\sum_i (\sqrt{p_i}(\sqrt{n}+\sqrt{m})+\sqrt{n\lambda_i} + \sqrt{m}))$, which is the same optimization problem solved in \cite{JKM12}, without log factors. This is $O(n\sqrt{l})$.


\section{Discussion and open questions}
\label{sec:open}

Some readers may wonder if the composition theorem could be avoided by using a standard argument about expected running times (or query complexity), which has the following form: Given $k$ Las Vegas algorithms with expected running times $t_1,\ldots,t_k$, running these algorithms in succession will yield an algorithm with expected running time $\sum_i t_i$ by the linearity of expectation. If we now terminate the algorithm after (say) 5 times its expected running time, then by Markov's inequality we have a bounded-error algorithm with worst-case running time $O(\sum_i q_i)$. However, to use this argument the individual algorithms need to be zero error. If the algorithms are merely bounded error, then the final answer may be incorrect even if one of the $k$ bounded-error algorithms errs. In our applications, oracle identification and Boolean matrix multiplicaiton, we use a subroutine to find the first marked 1 in a string. This algorithm has bounded error since it is too expensive to verify (with zero error) that a given 1 is indeed the first 1 in a string.

Our composition theorem only works for solutions of the filtered $\gamma_2$-norm SDP, not for quantum query complexity itself.  While this is sufficient for our application, it would be interesting to know if bounded-error quantum algorithms with input-dependent query complexities can be composed in general without incurring log factors.

While the query complexity of oracle identification in terms of $M$ and $N$ has been fully characterized, finding an optimal quantum algorithm for $\oip(\C)$ remains open.  The corresponding problem for classical query complexity is also open. 
It would also be interesting to study time-efficient oracle identification algorithms for specific sets $\C$, since none of the known algorithms, including ours, is known to be time efficient.


\section*{Acknowledgments}

I thank Andrew Childs and Ben Reichardt for helpful discussions, Seiichiro Tani for pointing me to Ref.~\cite{AIN+09}, and Andrew Childs and Ansis Rosmanis for comments on a preliminary draft.
This work was supported in part by NSERC, the Ontario Ministry of Research and Innovation, and the US ARO.


\bibliographystyle{amsalpha}
\bibliography{OIP}

\appendix

\section{Oracle identification lower bound}
\label{app:lb}

The main result, \thm{quantumOIP}, naturally has two parts. In this section we prove the lower bound: For any $N < M \leq 2^N$, $Q(\oip(M,N)) = \Omega\(\sqrt{\frac{N\log M}{\log({N}/{\log M})+1}}\)$.

We start with the following lemma, which follows from the proof of \cite[Theorem 2]{AIK+04}, and also appears as \cite[Theorem 5]{AIK+07}.

\begin{lemma}
\label{lem:threshold}
There exists a set of $N$-bit strings, $\C$, of size at most $M$, such that $Q(\oip(\C)) = \Omega(\sqrt{(N-k+1)k})$, for any $k$ that satisfies $\binom{N}{k-1} + \binom{N}{k} \leq M$. 
\end{lemma}

This can be shown using the original quantum adversary method of Ambainis \cite{Amb02}. First we prove a lower bound for the promise $k$-threshold problem, in which we have to decide if the input has Hamming weight $k-1$ or $k$ promised that one of these is the case. This problem has a lower bound of $\Omega(\sqrt{(N-k+1)k})$. Thus if we take $\C$ to be the set of all strings with Hamming weight $k-1$ or $k$, the oracle identification problem on this set is at least as hard as the promise $k$-threshold problem, which gives us the claimed lower bound in \lem{threshold}.

Now it suffices to prove the following lemma.

\begin{lemma}
\label{lem:binom}
For any $N < M \leq 2^N$, there exists a $k$ in  $\Omega\({\frac{\log M}{\log({N}/{\log M})+1}}\)$ such that $\binom{N}{k-1} + \binom{N}{k} \leq M$.
\end{lemma}

\begin{proof}
First note that if $M>2^{N/2}$, then $k=N/10$ satisfies the statement of the lemma, since $\binom{N}{k} \leq (Ne/k)^k \leq (10e)^{N/10} \leq 2^{\log(10e)N/10} < 2^{0.48N}$. In this range of $M$, $\Omega({\frac{\log M}{\log({N}/{\log M})+1}}) = \Omega(N)$.

For $M\leq 2^{N/2}$, let us choose  $k = c\frac{\log M}{\log({N}/{\log M})}$, for some constant $c<1$. In this range of $M$, this choice of $k$ is $\Omega({\frac{\log M}{\log({N}/{\log M})+1}})$.  Now we want $\binom{N}{k-1} + \binom{N}{k} \leq M$. Instead let us enforce that $\binom{N}{k} \leq M/2$ or $\log\binom{N}{k} / \log(M/2) \leq 1$. For convenience, let $m = \log M$.

We have $\frac{\log\binom{N}{k}}{m -1} \leq \frac{k\log(Ne/k)}{m -1} = c\frac{{m}\log(Ne/k)}{(\log({N}/{m}))(m-1)} = c \frac{m}{m - 1} \frac{\log(N/k)+\log e}{\log({N}/{m})}$. Since $m$ is large, $\frac{m}{m - 1} \leq 2$, so this expression is at most $2c \frac{\log(N/k)+\log e}{\log({N}/{m})}$, which is $2c \frac{\log(N/m)+\log(\log(N/m)/c)+\log e}{\log({N}/{m})}$, which is $2c \(1 + \frac{\log\log(N/m) + \log(1/c) +\log e}{\log (N/m)}\)$. Now since $N/m \geq 2$ by assumption, there is a choice for $c$ that makes this expression less than 1.
\end{proof}
Combining \lem{threshold} and \lem{binom} gives us $Q(\oip(M,N)) = \Omega\(\sqrt{\frac{N\log M}{\log({N}/{\log M})+1}}\)$.

\section{Proof of \texorpdfstring{\lem{opt}}{Lemma}}
\label{app:proofs}

\opt*
\begin{proof}
First we define two closely related optimization problems and show that their optimum values upper bound $C(M,N)$. Let $\alpha_1$ and $\alpha_2$ denote the optimum values of problem 1 and 2 respectively. We will show that $C(M,N) \leq \alpha_1 \leq \alpha_2$ and then upper bound $\alpha_2$ using the dual of problem 2. Let $n \defeq \ceil{\log N}$ and $m \defeq \ceil{\log M}$.
\begin{center}
  \begin{minipage}[t]{2.6in}
    \centerline{\underline{Problem 1 ($\alpha_1$)}}\vspace{-7mm}
    \begin{align*}
      \text{maximize:}\quad & \sum_{i=1}^{r} \sqrt{q_i}\\
      \text{subject to:}\quad & \sum_{i=1}^{r} q_i \leq 2N,\\
			& \prod_{i=1}^{r} q_i \leq M^2,\\
			& r \in [N],\\
			& 2 \leq q_i \leq 2N \quad (\text{for }i \in [r]).
    \end{align*}
  \end{minipage}
  \hspace*{13mm}
  \begin{minipage}[t]{2.6in}
    \centerline{\underline{Problem 2 ($\alpha_2$)}}\vspace{-7mm}
    \begin{align*}
      \text{maximize:}\quad & \sum_{k=1}^{n+2}  \sqrt{2^k} x_k\\
      \text{subject to:}\quad & \sum_{k=1}^{n+2} 2^k x_k \leq 4N,\\
			& \sum_{k=1}^{n+2} k x_k \leq 4m,\\
			& x_k \geq 0 \quad (\text{for }k \in [n+2]).
    \end{align*}
  \end{minipage}
\end{center}
Let $p_1, \ldots , p_r, r$ be an optimal solution of the problem in the statement of the lemma. Thus $C(M,N) = \sum_{i=1}^{r} \sqrt{p_i}$. Define $q_i = 2p_i$, for all $i \in [r]$. This is a feasible solution of problem 1, since $\sum_{i} p_i \leq N \Rightarrow \sum_{i} q_i \leq 2N$, and $\prod_{i} \max\{2,p_i\} \leq M$ gives us $\prod_{i} 2 \leq M$ and $\prod_{i} p_i \leq M$, which together yield $\prod_{i} 2p_i \leq M^2$. Finally  $\sum_{i} \sqrt{p_i} \leq \sum_{i} \sqrt{2p_i}$, which gives us $C(M,N) \leq \alpha_1$.

Now let $q_1, \ldots , q_r, r$ be an optimal solution of problem 1. Thus $\alpha_1 = \sum_{i=1}^{r} \sqrt{q_i}$. Define $x_k = |\{i:\ceil{\log q_i} = k\}|$. We claim that this is a feasible solution of problem 2. $\sum_{i} q_i \leq 2N \Leftrightarrow \sum_{i} 2^{\log q_i} \leq 2N$, which implies $\sum_{i} 2^{\ceil{\log q_i}} \leq 4N$. We can rewrite $\sum_{i} 2^{\ceil{\log q_i}}$ as $\sum_k 2^k x_k$, which gives us  $\sum_{k} 2^k x_k \leq 4N$. The next constraint $\prod_{i} q_i \leq M^2$ implies $\sum_{i} \log q_i \leq 2m$. Since each $q_i \geq 2$, the number of terms in this sum is at most $2m$, thus $\sum_{i} \ceil{\log q_i} \leq \sum_{i} (\log q_i+1) \leq 4m$. Again, $\sum_{i} \ceil{\log q_i}$ is the same as $\sum_k k x_k$, which gives us $\sum_{k} k x_k \leq 4m$. Finally $\alpha_1 = \sum_{i} \sqrt{q_i} = \sum_{i} \sqrt{2^{\log q_i}} \leq \sqrt{2^{\ceil{\log q_i}}} \leq \sum_k \sqrt{2^k}x_k \leq \alpha_2$.

Problem 2 is a linear program, which gives us an easy way to upper bound $\alpha_2$. For convenience, let $N' = 4N$, $n' = \ceil{\log N'} = n+2$, and $m' = 4m$. Let the optimum values of the following primal and dual linear programs be $\alpha$ and $\beta$ respectively. Clearly $\alpha_2 = \alpha$. By weak duality of linear programming, we have $\alpha \leq \beta$.

\begin{center}
  \begin{minipage}[t]{2.6in}
    \centerline{\underline{Primal ($\alpha$)}}\vspace{-7mm}
    \begin{align*}
      \text{maximize:}\quad & \sum_{k=1}^{n'}  \sqrt{2^k} x_k\\
      \text{subject to:}\quad & \sum_{k=1}^{n'} 2^k x_k \leq N',\\
			& \sum_{k=1}^{n'} k x_k \leq m',\\
			& x_k \geq 0 \quad (\text{for }k \in [n']).
    \end{align*}
  \end{minipage}
  \hspace*{13mm}
  \begin{minipage}[t]{2.6in}
    \centerline{\underline{Dual ($\beta$)}}\vspace{-7mm}
    \begin{align*}
      \text{minimize:}\quad & N' y + m' z\\
      \text{subject to:}\quad & 2^k y + kz \geq \sqrt{2^k}, \quad (\text{for }k \in [n'])\\
			& y,z \geq 0.
    \end{align*}
  \end{minipage}
\end{center}

For convenience, define $d = \log(2N'/m')=\log(2N/m)$, which satisfies $d\geq 1$ since $m\leq N$. We can use any dual feasible solution to upper bound $\beta$.  Let $y= \sqrt{\frac{1}{2^d d}}$ and $z = \sqrt{\frac{2^d}{d}}$. Thus $\beta \leq N'y + m'z \leq 2\sqrt{2}\sqrt{\frac{N'm'}{\log(N'/m')+1}} = O\(\sqrt{\frac{Nm}{\log({N}/{m})+1}}\)$.

Let us check the constraints: Clearly $y,z \geq 0$; the other constraints require that
\[
\sqrt{\frac{2^k}{d2^d}} + \sqrt{\frac{k^22^d}{2^k d}} \geq {1}
\]
for all $k\geq 1$ and $d\geq 1$. Using $a+b \geq 2\sqrt{ab}$, the left-hand side of this equation is greater than $2k/d$. Thus the inequality clearly holds for $k \geq d$ (and even $k \geq d/2$). 

Now suppose $1\leq k\leq d$. Let us show that the second term $\sqrt{\frac{k^22^d}{2^k d}}$ is large enough. Since $\frac{k^2}{2^k}$ is concave in this range, the minimum is achieved at either $k=1$ or $k=d$.   For $k=1$, the second term becomes $\sqrt{2^d/d}$, and for $k=d$, the second term evaluates to $\sqrt{d}$. Both of which are at least 1 when $d \geq 1$.

Since the solution is feasible, we get $C(M,N) \leq \beta = O\(\sqrt{\frac{Nm}{\log({N}/{m})+1}}\)$. Finally, note that in the problem in the statement of the lemma, $\prod_{i} \max\{2,p_i\} \leq M$ forces $\sum_i p_i \leq M$, which also implies $p_i \leq M$ and $r \in M$.  Thus we may simply substitute $N$ with $M$ to get another valid upper bound. This gives us $C(M,N) = O(\sqrt{M})$.
\end{proof}

\end{document}